\documentclass[preprint,12pt,3p]{elsarticle}

\usepackage{amssymb}

\usepackage{amsthm}
\usepackage{graphicx}
\usepackage{amsmath}
\usepackage{hyperref}

\usepackage{pgf,tikz}

\usetikzlibrary{arrows}
\usetikzlibrary{shapes}
\usepackage{multirow}
\usepackage{latexsym}
\usepackage{setspace}
\PassOptionsToPackage{boxed,section}{algorithm}
\usepackage{algorithm}
\usepackage{algorithmic}
\usepackage{enumerate}
\usepackage{amsmath}			
\usepackage{amssymb}			
\usepackage{url}
\usepackage{setspace}
\usepackage{tikz}
\usetikzlibrary{arrows}
\usetikzlibrary{shapes}
\usetikzlibrary{decorations.markings}
\usepackage{geometry}
\usepackage{bbm}

\usepackage{paralist}

\newtheorem{proposition}{\em Proposition}
\newtheorem{theorem}{\em Theorem}
\newtheorem{conjecture}{\em Conjecture}
\newtheorem{definition}{\em Definition}
\newtheorem{lemma}{\em Lemma}

\newtheorem{remark}{\em Remark}

\journal{Sample Journal}

\begin{document}

\begin{frontmatter}

\title{Normal $6$-edge-colorings of some bridgeless cubic graphs}

\address[label1]{Dipartimento di Informatica,
Universita degli Studi di Verona, Strada le Grazie 15, 37134 Verona, Italy}

\cortext[cor1]{Corresponding author}

\author[label1]{Giuseppe Mazzuoccolo\corref{cor1}}

\ead{giuseppe.mazzuoccolo@univr.it}

\author[label1,label5]{Vahan Mkrtchyan}
\address[label5]{Gran Sasso Science Institute,
School of Advanced Studies, L'Aquila, Italy}
\ead{vahan.mkrtchyan@gssi.it}

\begin{abstract}
In an edge-coloring (proper) of a cubic graph, an edge is poor or rich, if the set of colors assigned to the edge and the four edges adjacent it, has exactly three or exactly five distinct colors, respectively. An edge is normal in an edge-coloring if it is rich or poor in this coloring. A normal $k$-edge-coloring of a cubic graph is an edge-coloring with $k$ colors such that each edge of the graph is normal. We denote by $\chi'_{N}(G)$ the smallest $k$, for which $G$ admits a normal $k$-edge-coloring. Normal edge-colorings were introduced by Jaeger in order to study his well-known Petersen Coloring Conjecture. It is known that proving $\chi'_{N}(G)\leq 5$ for every bridgeless cubic graph is equivalent to proving Petersen Coloring Conjecture. Moreover, Jaeger was able to show that it implies classical conjectures like Cycle Double Cover Conjecture and Berge-Fulkerson Conjecture. Recently, two of the authors were able to show that any simple cubic graph admits a normal $7$-edge-coloring, and this result is best possible. In the present paper, we show that any claw-free bridgeless cubic graph, permutation snark, tree-like snark admits a normal $6$-edge-coloring. Finally, we show that any bridgeless cubic graph $G$ admits a $6$-edge-coloring such that at least $\frac{7}{9}\cdot |E|$ edges of $G$ are normal.
\end{abstract}


\end{frontmatter}


\section{Introduction}
\label{sec:intro}

The Petersen Coloring Conjecture in graph theory asserts that the edge-set of every bridgeless cubic graph $G$ can be colored by using as set of colors the edge-set of the Petersen graph $P_{10}$ in such a way that adjacent edges of $G$ receive as colors adjacent edges of $P_{10}$.
The conjecture is well-known and it is largely considered hard to prove since it implies classical conjectures in the field such as Cycle Double Cover Conjecture and Berge-Fulkerson Conjecture (see \cite{Fulkerson,Jaeger1985,Zhang1997}).
In \cite{Jaeger1985}, Jaeger introduced an equivalent formulation of the Petersen Coloring Conjecture. More precisely, he showed that a bridgeless cubic graph is a counterexample to this conjecture, if and only if, it does not admit a normal edge-coloring (see Definitions \ref{def:poorrich} and \ref{def:normal} in Section \ref{sec:intro}) with at most $5$ colors. Let $\chi'_{N}(G)$ denote the normal chromatic index of $G$, that is, the minimum number of colors in a normal edge-coloring of $G$. In this terms, Petersen Coloring Conjecture is equivalent to saying that every bridgeless cubic graph has normal chromatic index at most 5. As far as we know, the best known upper bound for an arbitrary bridgeless cubic graph is $7$ (see \cite{Bilkova12,Normal7flows}). 
 There exist examples of simple cubic graphs (not bridgeless) with normal chromatic index $7$. On the other hand, in \cite{Normal7flows} it is shown that any simple cubic graph admits a normal $7$-edge-coloring. Let us recall that a weaker upper bound for an arbitrary simple cubic graph was proved in \cite{Bilkova12}. One may wonder whether the upper bound $7$ can be improved in other interesting subclasses of cubic graphs. Due to Conjecture \ref{conj:5NormalConj}, bridgeless cubic graphs form the first important case that one needs to study. Since obtaining an upper bound five for $\chi'_{N}(G)$ in this class is hard (Conjecture \ref{conj:5NormalConj}), one can try to show a weaker upper bound for $\chi'_{N}(G)$, that is, six (Conjecture \ref{conj:6normalBridgelessCubic}). Unfortunately, we are unable to prove this statement in general. This is the main reason why in this paper we consider some subclasses of bridgeless cubic graphs where we verify the statement, hence obtain partial results towards it. In subsection \ref{subsec:ClawFree} we verify Conjecture \ref{conj:6normalBridgelessCubic} in the class of claw-free bridgeless cubic graphs. Then, in subsections \ref{subsec:PermSnarks} and \ref{subsec:TreeLikeSnarks} we verify the conjecture for permutation snarks and treelike snarks, respectively. Finally, in subsection \ref{subsec:NormalEdges} we give a non-trivial lower bound for the number of normal edges in a 6-edge-coloring of a bridgeless cubic graph. 

Now, let us introduce the main definitions and notions used in the paper in detail. Graphs considered in this paper are finite and undirected. They do not contain loops, though they may contain parallel edges. A graph is simple if it contains no parallel edge.



For a graph $G$, let $V(G)$ and $E(G)$ be the set of vertices and edges of $G$, respectively. Moreover, let $\partial_{G}(v)$ be the set of edges of $G$ that are incident to the vertex $v$ of $G$. A subgraph $H$ of $G$ is even, if every vertex of $H$ has even degree in $H$. A matching of $G$ is a set of edges of $G$ such that any two of them do not share a vertex. A matching of $G$ is perfect, if it contains $\frac{|V(G)|}{2}$ edges. For a positive integer $k$, a $k$-factor of $G$ is a spanning $k$-regular subgraph of $G$. Observe that if $G$ is a cubic graph, then $F$ is a $1$-factor of $G$, if and only if the set $E(G)\setminus E(F)$ is an edge-set of a $2$-factor of $G$. These $1$-factor and $2$-factor are said to be complementary.

Let $G$ and $H$ be two cubic graphs. If there is a mapping $\phi:E(G)\rightarrow E(H)$, such that for each $v\in V(G)$ there is $w\in V(H)$ such that $\phi(\partial_{G}(v)) = \partial_{H}(w)$, then $\phi$ is called an $H$-coloring of $G$. If $G$ admits an $H$-coloring, then we will write $H
\prec G$. It can be easily seen that if $H\prec G$ and $K\prec H$, then $K\prec G$. In other words, $\prec$ is a transitive relation defined on the set of cubic graphs.


    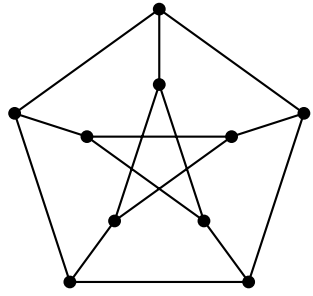
\begin{figure}[ht]
	\begin{center}
	\begin{tikzpicture}[style=thick]
\draw (18:2cm) -- (90:2cm) -- (162:2cm) -- (234:2cm) --
(306:2cm) -- cycle;
\draw (18:1cm) -- (162:1cm) -- (306:1cm) -- (90:1cm) --
(234:1cm) -- cycle;
\foreach \x in {18,90,162,234,306}{
\draw (\x:1cm) -- (\x:2cm);
\draw[fill=black] (\x:2cm) circle (2pt);
\draw[fill=black] (\x:1cm) circle (2pt);
}
\end{tikzpicture}
	\end{center}
	\caption{The graph $P_{10}$.}\label{fig:Petersen10}
\end{figure}

Let $P_{10}$ be the well-known Petersen graph (Figure \ref{fig:Petersen10}). The Petersen coloring conjecture of Jaeger states:
\begin{conjecture}\label{conj:P10conj} (Jaeger, 1988 \cite{Jaeger1988}) For any bridgeless cubic
graph $G$, we have $P_{10} \prec G$.
\end{conjecture}

Note that the Petersen graph is the only bridgeless cubic graph that can color all bridgeless cubic graphs \cite{Mkrt2013}. The conjecture is clearly difficult to prove, since it implies the classical Berge-Fulkerson conjecture \cite{Fulkerson,Seymour} and (5,2)-cycle-cover conjecture \cite{Celmins1984,Preiss1981}.

%

A $k$-edge-coloring of a graph $G$ is an assignment of colors $\{1,...,k\}$ to edges of $G$, such that adjacent edges receive different colors. If $c$ is an edge-coloring of $G$, then for a vertex $v$ of $G$, let $S_{c}(v)$ be the set of colors that the edges incident to $v$ receive. 

\begin{definition}\label{def:poorrich}
Let $uv$ be an edge of a cubic graph $G$ and $c$ be an edge-coloring of $G$. The edge $uv$ is {\bf poor} if $|S_{c}(u)\cup S_{c}(v)|=3$ and {\bf rich} if $|S_{c}(u)\cup S_{c}(v)|=5$. An edge is normal with respect to $c$ if it is poor or rich.
\end{definition}

Edge-colorings having only poor edges are trivially $3$-edge-colorings of $G$. Also edge-colorings having only rich edges have been considered before, and they are called strong edge-colorings \cite{Andersen1992}. In this paper, we will focus on the case when all edges must be normal.

\begin{definition}\label{def:normal}
An edge-coloring $c$ of a cubic graph is {\bf normal}, if any edge is normal with respect to $c$. 
\end{definition} 

It is straightforward that an edge coloring which assigns a different color to every edge of a simple cubic graph is normal since all edges are rich. Hence, we can define the normal chromatic index of a simple cubic graph $G$, denoted by $\chi'_{N}(G)$, as the smallest $k$, for which $G$ admits a normal $k$-edge-coloring. In \cite{Jaeger1985}, Jaeger has shown that:

\begin{proposition}\label{prop:JaegerNormalColor}(Jaeger, \cite{Jaeger1985}) If $G$ is a cubic graph, then $P_{10}\prec G$, if and only if $G$ admits a normal $5$-edge-coloring.
\end{proposition} This implies that Conjecture \ref{conj:P10conj} can be stated as follows:

\begin{conjecture}\label{conj:5NormalConj} For any bridgeless cubic graph $G$, $\chi'_{N}(G)\leq 5$.
\end{conjecture} Observe that Conjecture \ref{conj:5NormalConj} is trivial for $3$-edge-colorable cubic graphs. This is true because in any $3$-edge-coloring $c$ of a cubic graph $G$ any edge $e$ is poor, hence $c$ is a normal edge-coloring of $G$. Thus non-$3$-edge-colorable cubic graphs are the main obstacle to prove Conjecture \ref{conj:5NormalConj}. Note that Conjecture \ref{conj:5NormalConj} is verified for some non-$3$-edge-colorable bridgeless cubic graphs in \cite{HaggSteff2013}. Finally, note that in \cite{SamalElecNotes} the percentage of edges of a bridgeless cubic graph, which can be made normal in a 5-edge-coloring, is investigated.

In this paper, we focus on the problem of finding better upper bound for $\chi'_{N}(G)$ in the class of bridgeless cubic graphs. Since all simple cubic graphs admit a normal $7$-edge-coloring \cite{Normal7flows}, and proving $\chi'_{N}(G)\leq 5$ is hard (Conjecture \ref{conj:5NormalConj}), we focus on obtaining an upper bound $6$ for some bridgeless cubic graphs (Conjecture \ref{conj:6normalBridgelessCubic}). Terms and concepts that we do not define can be found in standard books like \cite{West}.

\section{Some Auxiliary Results}
\label{sec:aux}

In this section, we present some results that will be used later. Let us recall some basic terminology of flow theory which will be one of the techniques used in order to prove our results.

Let $A$ be an Abelian group with respect to $+$, and let $0$ be the unit element of $A$. If $G$ is a graph, then we say that $G$ admits a nowhere-zero $A$-flow, if there is an orientation $D$ of edges of $G$ and a mapping $\phi:E(G)\rightarrow A\setminus \{0\}$, such that for any vertex $v$ of $G$
\begin{equation*}
    \sum_{e\in \partial^{+}(v)}\phi(e)=\sum_{e\in \partial^{-}(v)}\phi(e).
\end{equation*} Here $\partial^{+}(v)$ and $\partial^{-}(v)$ denote the set of edges of $G$ leaving and entering $v$, respectively. 

It can be shown that if a graph $G$ admits a nowhere-zero $A$-flow with respect to some orientation $D$, then it admits a nowhere-zero $A$-flow with respect to any orientation. Hence, we can speak of $G$ having a nowhere-zero $A$-flow without specifying the orientation.

In the following two classical theorems of Jaeger, $\mathbb{Z}_{2}$ denotes the cyclic group of order $2$, and $\times$ is the direct product of groups. In what follows, we will denote, as usual, the direct product $\mathbb{Z}_2 \times \mathbb{Z}_2$ ($\mathbb{Z}_2 \times \mathbb{Z}_2 \times \mathbb{Z}_2$) by $\mathbb{Z}^2_2$ ($\mathbb{Z}^3_2$) (i.e. elementary abelian groups of order $4$ and $8$).

\begin{theorem}\label{thm:Jaeger8flow} (Jaeger, \cite{Jaeger1975,Jaeger1979}) Any bridgeless graph admits a nowhere-zero $\mathbb{Z}^3_{2}$-flow.
\end{theorem}

\begin{theorem}\label{thm:Jaeger4flow} (Jaeger, \cite{Jaeger1975,Jaeger1979}) Any 4-edge-connected graph admits a nowhere-zero $\mathbb{Z}^2_{2}$-flow.
\end{theorem}

Finally, we will need the following well-known consequence of Edmonds' Theorem:

\begin{theorem}\label{thm:KaiserKral} Any bridgeless cubic graph $G$ admits a perfect matching $F$, such that $F$ intersects any 3-edge-cut of $G$ in a single edge.
\end{theorem}

\section{The main results}
\label{sec:Main}


In this section, we present our main results. They deal with the following question that was also asked by Robert \v{S}\'{a}mal:

\begin{conjecture}
\label{conj:6normalBridgelessCubic} Let $G$ be a bridgeless cubic graph. Then $\chi'_{N}(G)\leq 6$.
\end{conjecture}

Since we are  unable to prove this conjecture, we verify it in some subclasses of bridgeless cubic graphs, hence obtain partial results towards it.


\subsection{Claw-free cubic graphs}
\label{subsec:ClawFree}

We will need some results on claw-free simple cubic graphs. Recall that a graph $G$ is claw-free, if it does not contain four vertices, such that the subgraph of $G$ induced on these vertices is isomorphic to $K_{1,3}$. It turns out that this class is interesting in this context since the restriction of Conjecture \ref{conj:5NormalConj} for claw-free cubic graphs implies its truth in general. 
In order to see this, let $G$ be any bridgeless cubic graph. Consider a bridgeless cubic graph $G_{\Delta}$ obtained from $G$ by replacing any vertex of $G$ with a triangle. Observe that $G_{\Delta}$ is claw-free. Now let $c$ be a normal 5-edge-coloring of $G_{\Delta}$. Take any triangle $T$ in $G_{\Delta}$. If we assume that the edges of $T$ are colored with colors 1, 2 and 3 in $c$, then one of three edges adjacent to a vertex $T$ and lying outside $T$ must be colored with 1, 2 or 3. Now it is not hard to see that this implies that all six edges incident to vertices of $T$ are colored with 1, 2 or 3. If we contract all the triangles of $G_{\Delta}$ that correspond to vertices of $G$ and consider the restriction of $c$ to $G$, then clearly it will be a normal 5-edge-coloring of $G$.

In this section, we show that $\chi'_{N}(G)\leq 6$ for claw-free bridgeless cubic graphs. In \cite{ChudSeyClawFreeChar}, arbitrary claw-free graphs are characterized. In \cite{sang-il_oum:2011}, Oum has characterized simple, claw-free bridgeless cubic graphs. In order to formulate Oum's result, we need some definitions. In a claw-free simple cubic graph $G$ any vertex belongs to one, two, or three triangles. If a vertex $v$ belongs to three triangles of $G$, then the component of $G$ containing $v$ is isomorphic to $K_4$ (Figure \ref{fig:K4}). An induced subgraph of $G$ that is isomorphic to $K_4-e$ is called a diamond \cite{sang-il_oum:2011}. It can be easily checked that in a claw-free cubic graph no two diamonds intersect.  

\begin{figure}[!htbp]
\begin{center}
\begin{tikzpicture}[scale=0.35]

 \tikzstyle{every node}=[circle, draw, fill=black!,
                        inner sep=0pt, minimum width=4pt]
  
  \node[circle,fill=black,draw] at (-2,2) (n2) {};

   \node[circle,fill=black,draw] at (-2,-2) (n3) {};

    \node[circle,fill=black,draw] at (-6,2) (n4) {};
    
    \node[circle,fill=black,draw] at (-6,-2) (n5) {};

 \draw (n2)--(n3);
 
 \draw (n2)--(n4);
 \draw (n3)--(n4);
 \draw (n3)--(n5);
 \draw (n2)--(n5);
 
 \draw (n4)--(n5);

\end{tikzpicture}
\end{center}
\caption{The graph $K_4$.} \label{fig:K4}
\end{figure}
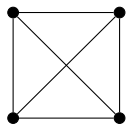

A string of diamonds of $G$ is a maximal sequence $F_{1},...,F_{k}$ of diamonds, in which $F_{i}$ has a vertex adjacent to a vertex of $F_{i+1}$, $1\leq i \leq k-1$.  A string of diamonds has exactly two vertices of degree two, which are called the head and the tail of the string. Replacing an edge $e = uv$ with a string of diamonds with the head $x$ and the tail $y$ is to remove $e$ and add edges $(u,x)$ and $(v,y)$.

If $G$ is a connected claw-free simple cubic graph such that each vertex lies in a diamond, then $G$ is called a ring of diamonds. It can be easily checked that each vertex of a ring of diamonds lies in exactly one diamond. As in \cite{sang-il_oum:2011}, we require that a ring of diamonds contains at least two diamonds.

\begin{proposition}\label{prop:OumClawfreebridgelessCharac} (Oum, \cite{sang-il_oum:2011}) $G$ is a connected claw-free simple bridgeless cubic graph, if and only if
\begin{enumerate}
    \item [(1)] $G$ is isomorphic to $K_4$, or
    
    \item [(2)] $G$ is a ring of diamonds, or
    
    \item [(3)] there is a connected bridgeless cubic graph $H$, such that $G$ can be obtained from $H$ by replacing some edges of $H$ with strings of diamonds, and by replacing any vertex of $H$ with a triangle.
\end{enumerate}
\end{proposition}

We will need some additional definitions. Let $T$ be a triangle in a cubic graph $G$ such that each edge of $T$ is of multiplicity one. If $e$ is an edge of $T$, then let $f$ be the edge of $G$ that is incident to a vertex of $T$ and is not adjacent to $e$. The edges $e$ and $f$ will be called opposite. We prove the following lemma:

\begin{lemma}\label{lem:HperfectMatching} Let $H$ be a cubic graph containing a perfect matching, and let $H_{\Delta}$ be obtained from $H$ by replacing every vertex of $H$ by a triangle. Then $\chi'_{N}(H_{\Delta})\leq 6$.
\end{lemma}

\begin{proof} Let $F$ be a perfect matching of $H$, and let $\overline{F}$ be the $2$-factor of $H$ that is complementary to $F$. Since $H_{\Delta}$ is obtained from $H$ by replacing each vertex of $H$ with a triangle, with abuse of notation, we will always refer to the edges of $H$ as a subset of the edges of $H_{\Delta}$. Hence, we see $F$ as a matching (not perfect) of $H_{\Delta}$. We denote by $F'$ the matching of $H_{\Delta}$ consisting of all edges of the added triangles which are opposite to an edge of $F$. Note that $F \cup F'$ is a perfect matching of $H_{\Delta}$, and its complement $\overline{F \cup F'}$ is a $2$-factor of $H_{\Delta}$.

First, we color the edges of $H_{\Delta}$ in $F$ with color $1$. Now, let $C$ be a cycle of $\overline{F}$. The edges of $C$ belong to the edges of a unique cycle $C'$ of $\overline{F \cup F'}$. Moreover, by construction, the length of $C'$ is exactly three times the length of $C$. Hence, we have only two cases according to the parity of the cycle $C$: either the length of $C'$ is $6l$ or $6l+3$, for an arbitrary positive integer $l \geq 1$.\\

Case $6l$: Color all edges in $F'$ with two ends in $C$ with color $4$. Color edges of $C$ in the order by repeating $l$ times the sequence of colors $2,5,2,3,6,3$, in such a way that all edges of the added triangles of $H_{\Delta}$ receive colors $2$ and $3$.  On Figure \ref{fig:4CycleExample}, the coloring is presented when $l=2$.\\

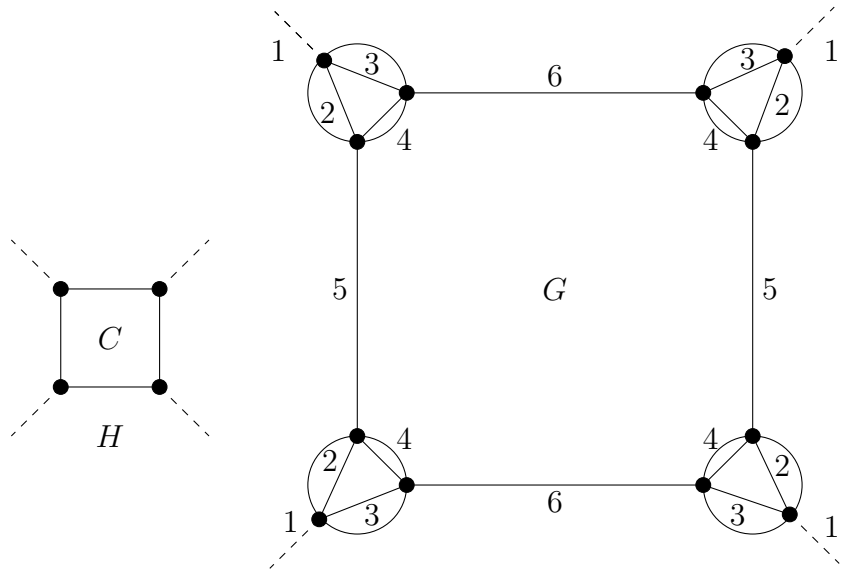
\begin{figure}[!htbp]
\begin{center}
\begin{tikzpicture}[scale=0.65]


 \draw[fill=black] (0,0) circle [radius=0.15cm] ;

  \draw[fill=black] (2,0) circle [radius=0.15cm] ;

   \draw[fill=black] (2,2) circle [radius=0.15cm] ;

    \draw[fill=black] (0,2) circle [radius=0.15cm] ;
    
 \draw (0,0) -- (2,0);
 
 \draw (2,0) -- (2,2);
 
 \draw (2,2) -- (0,2);
 
 \draw (0,2) -- (0,0);
 
 \draw[dashed] (0,0) -- (-1,-1);
 \draw[dashed] (2,0) -- (3,-1);
 \draw[dashed] (2,2) -- (3,3);
 \draw[dashed] (0,2) -- (-1,3);
 
  \node at (1, 1) (vC) {$C$};
  


\draw (6,-2) circle [radius=1cm] ;
\coordinate (a1) at (5.23,-2.7);
\coordinate (a11) at (4.23,-3.7);
\coordinate (b1) at (6,-1);
\coordinate (c1) at (7,-2);

\draw[fill=black] (a1) circle [radius=0.15cm] ;

 \draw[fill=black] (b1) circle [radius=0.15cm] ;
 
 \draw[fill=black] (c1) circle [radius=0.15cm] ;
 
 \draw (a1)--(b1)--(c1)--(a1);
 \draw[dashed] (a1) -- (a11);
 
 
 \node at (6.3, -2.6)  {$3$};
 
  \node at (5.45, -1.5)  {$2$};

  \node at (4.65, -2.75)  {$1$};
  \node at (6.95, -1.05)  {$4$};
  
  
\coordinate (a2) at (14.75,-2.6);
\coordinate (a22) at (15.75,-3.6);
\coordinate (b2) at (13,-2);
\coordinate (c2) at (14,-1);

\draw (14,-2) circle [radius=1cm] ;

\draw[fill=black] (a2) circle [radius=0.15cm] ;

\draw[fill=black] (b2) circle [radius=0.15cm] ;

 \draw[fill=black] (c2) circle [radius=0.15cm] ;
 
  \draw (a2)--(b2)--(c2)--(a2);
 \draw[dashed] (a2) -- (a22);
 
 \node at (13.7, -2.6)  {$3$};
  \node at (14.6, -1.6)  {$2$};
  \node at (15.6, -2.85)  {$1$};
  \node at (13.15, -1.05)  {$4$};

   
 \coordinate (a3) at (14.65,6.75);
\coordinate (a33) at (15.65,7.75);
\coordinate (b3) at (14,5);
\coordinate (c3) at (13,6);

    \draw (14,6) circle [radius=1cm] ;
    
    \draw[fill=black] (a3) circle [radius=0.15cm] ;
    
 \draw[fill=black] (b3) circle [radius=0.15cm] ;
 
 \draw[fill=black] (c3) circle [radius=0.15cm] ;

  \draw (a3)--(b3)--(c3)--(a3);
 \draw[dashed] (a3) -- (a33);
 
 \node at (14.6, 5.75)  {$2$};
 \node at (13.9, 6.7)  {$3$};
  \node at (15.6, 6.85)  {$1$};
  \node at (13.15, 5.05)  {$4$};
  
  
 \coordinate (a4) at (5.33,6.66);
\coordinate (a44) at (4.33,7.66);
\coordinate (b4) at (7,6);
\coordinate (c4) at (6,5);

   \draw (6,6) circle [radius=1cm] ;
   
   \draw[fill=black] (a4) circle [radius=0.15cm] ;
   
   \draw[fill=black] (b4) circle [radius=0.15cm] ;
   
 \draw[fill=black] (c4) circle [radius=0.15cm] ;
 
 \draw (a4)--(b4)--(c4)--(a4);
 \draw[dashed] (a4) -- (a44);
 
 \node at (5.4, 5.6)  {$2$};
 \node at (6.3, 6.6)  {$3$};
  \draw[dashed] (a4) -- (a44);
  \node at (4.4, 6.85)  {$1$};
  \node at (6.95, 5.05)  {$4$};
  
  
  \draw (c1) -- (b2);
    \node at (10, -2.35)  {$6$};
  
  \draw (c2) -- (b3);
  \node at (14.35, 2)  {$5$};
  
  \draw (c3) -- (b4);
  \node at (10, 6.35)  {$6$};
  
  \draw (b1) -- (c4);
   \node at (5.65, 2)  {$5$};

  
  \node at (10, 2) (vC) {$G$};
  \node at (1, -1) (vH) {$H$};

\end{tikzpicture}
\end{center}
\caption{The coloring of $C$ in the case $6l$, when $l=2$. The edges of $F$ are dashed.} \label{fig:4CycleExample}
\end{figure}

Case $6l+3$: Consider nine consecutive edges of $C$ in such a way that the first and the last of them are edges of the added triangles of $H_{\Delta}$. Color them in the order with the sequence of colors $2,5,3,2,4,3,2,6,3$.
If $l>1$, color the remaining $6(l-1)$ edges of $C$ in the order by repeating $l-1$ times the sequence of colors $2,5,2,3,6,3$. 
Finally, color all edges in $F'$ with two ends in $C$ with the unique color in $\{4,5,6\}$ which gives a proper coloring of $H_{\Delta}$. It is easy to see that such a color always exists and it is uniquely determined. 
On Figure \ref{fig:5CycleExample}, the coloring is presented when $l=2$.\\

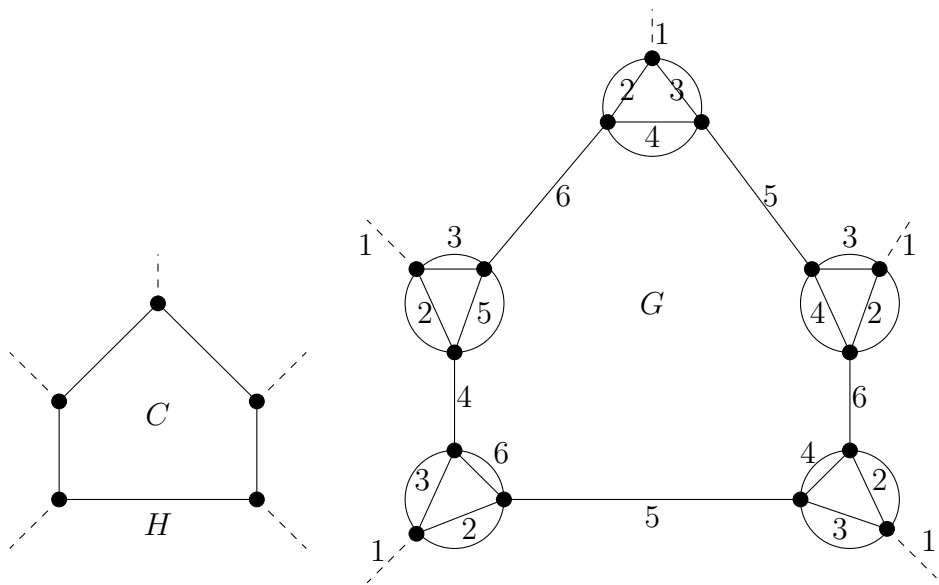
\begin{figure}[!htbp]
\begin{center}
\begin{tikzpicture}[scale=0.65]


 \draw[fill=black] (0,2) circle [radius=0.15cm] ;

  \draw[fill=black] (2,-2) circle [radius=0.15cm] ;

   \draw[fill=black] (-2,-2) circle [radius=0.15cm] ;
   
   \draw[fill=black] (-2,0) circle [radius=0.15cm] ;
   
   \draw[fill=black] (2,0) circle [radius=0.15cm] ;

 \draw (0,2) -- (-2,0);
 \draw (0,2) -- (2,0);
 
 \draw (-2,-2) -- (-2,0);
 \draw (2,-2) -- (2,0);
 
 \draw (-2,-2) -- (2,-2);

 \draw[dashed] (0,2) -- (0,3);
 \draw[dashed] (2,-2) -- (3,-3);
 \draw[dashed] (-2,-2) -- (-3,-3);
 \draw[dashed] (-2,0) -- (-3,1);
 \draw[dashed] (2,0) -- (3,1);

  \node at (0, -0.25) (vC) {$C$};
  


\draw (6,-2) circle [radius=1cm] ;
\coordinate (a1) at (5.23,-2.7);
\coordinate (a11) at (4.23,-3.7);
\coordinate (b1) at (6,-1);
\coordinate (c1) at (7,-2);

\draw[fill=black] (a1) circle [radius=0.15cm] ;

 \draw[fill=black] (b1) circle [radius=0.15cm] ;
 
 \draw[fill=black] (c1) circle [radius=0.15cm] ;
 
 \draw (a1)--(b1)--(c1)--(a1);
 \draw[dashed] (a1)--(a11);
 
 \node at (6.3, -2.6)  {$2$};
  \node at (5.35, -1.6)  {$3$};
  
  \node at (4.45, -3.05)  {$1$};
  \node at (6.95, -1.05)  {$6$};
  
  
  \draw (14,-2) circle [radius=1cm] ;
  
\coordinate (a2) at (14.75,-2.6);
\coordinate (a22) at (15.75,-3.6);
\coordinate (b2) at (13,-2);
\coordinate (c2) at (14,-1);

\draw[fill=black] (a2) circle [radius=0.15cm] ;

 \draw[fill=black] (b2) circle [radius=0.15cm] ;
 
 \draw[fill=black] (c2) circle [radius=0.15cm] ;
 
  \draw (a2)--(b2)--(c2)--(a2);
 \draw[dashed] (a2)--(a22);
 
 \node at (13.8, -2.6)  {$3$};
  \node at (14.6, -1.6)  {$2$};
 
  \node at (15.6, -2.85)  {$1$};
  \node at (13.15, -1.05)  {$4$};
  

\draw (6,2) circle [radius=1cm] ;
\coordinate (a4) at (5.23,2.7);
\coordinate (a44) at (4.23,3.7);
\coordinate (b4) at (6,1);
\coordinate (c4) at (6.6,2.7);

\draw[fill=black] (a4) circle [radius=0.15cm] ;

 \draw[fill=black] (b4) circle [radius=0.15cm] ;
 
 \draw[fill=black] (c4) circle [radius=0.15cm] ;
 
 \draw (a4)--(b4)--(c4)--(a4);
 \draw[dashed] (a4)--(a44);
 
 \node at (6.6, 1.8)  {$5$};
  \node at (5.4, 1.8)  {$2$};
  
  \node at (4.2, 3.2)  {$1$};
  \node at (6, 3.35)  {$3$};
 

\draw (14,2) circle [radius=1cm] ;

\coordinate (a5) at (14.6,2.7);
\coordinate (a55) at (15.23,3.7);
\coordinate (b5) at (13.23,2.7);
\coordinate (c5) at (14,1);

\draw[fill=black] (a5) circle [radius=0.15cm] ;

 \draw[fill=black] (b5) circle [radius=0.15cm] ;
 
 \draw[fill=black] (c5) circle [radius=0.15cm] ;
 
 \draw (a5)--(b5)--(c5)--(a5);
 \draw[dashed] (a5)--(a55);
 
\node at (14, 3.35)  {$3$};
 \node at (14.5, 1.8)  {$2$};
  
  \node at (15.2, 3.2)  {$1$};
 \node at (13.35, 1.8)  {$4$};
  
  
  \draw (10,6) circle [radius=1cm] ;
  
  \coordinate (a3) at (10,7);
\coordinate (a33) at (10,8);
\coordinate (b3) at (11,5.7);
\coordinate (c3) at (9.1,5.7);
  
  \draw[fill=black] (a3) circle [radius=0.15cm] ;
  \draw[fill=black] (b3) circle [radius=0.15cm] ;
  \draw[fill=black] (c3) circle [radius=0.15cm] ;
  
   \draw (a3)--(b3)--(c3)--(a3);
 \draw[dashed] (a3)--(a33);

  \node at (10, 5.4)  {$4$};
  
  \node at (9.5, 6.35)  {$2$};
  
  \node at (10.5, 6.35)  {$3$};

  \node at (10.2, 7.5)  {$1$};

  
  \draw (7,-2) -- (13,-2);
  \node at (10, -2.35)  {$5$};
  
  \draw (b4) -- (b1);
  \node at (6.2, 0.1)  {$4$};
  
  \draw (c2) -- (c5);
  \node at (14.2, 0.1)  {$6$};
  
  \draw (c4) -- (c3);
  
  \draw (b5) -- (b3);
  
  \node at (8.2, 4.2)  {$6$};
  
  \node at (12.4, 4.2)  {$5$};
  
  \node at (10, 2) (vC) {$G$};
  \node at (0, -2.5) (vH) {$H$};

\end{tikzpicture}
\end{center}
\caption{The coloring of $C$ in the case $6l+3$, when $l=2$. The edges of $F$ are dashed.} \label{fig:5CycleExample}
\end{figure}

\medskip 

It is not hard to see that the described coloring is a normal $6$-edge-coloring of $H_{\Delta}$. The proof is complete.
\end{proof}

We are ready to obtain the main result of this section.

\begin{theorem}\label{thm:6clawfreebridgeless} If $G$ is a claw-free bridgeless cubic graph, then $\chi'_{N}(G)\leq 6$. 
\end{theorem}

\begin{proof} We prove the theorem by induction on $|V(G)|$. If $|V(G)|=2$, then $G$ is $3$-edge-colorable, hence $\chi'_{N}(G)\leq 3$. Assume that the theorem is true for all claw-free bridgeless cubic graphs $G$ with $|V(G)|<n$, and let us consider a claw-free bridgeless cubic graph $G$ with $|V(G)|=n\geq 4$. Without loss of generality we can assume that $G$ is connected, otherwise the statement follows from inductive hypothesis for components of $G$.

First assume that $G$ contains two vertices $u$ and $v$ that are joined by two parallel edges. Let $u'$ and $v'$ be the neighbors of $u$ and $v$, respectively, that are different from $v$ and $u$. Consider the cubic graph $G'$ defined as follows:
\begin{equation*}
    G'=(G-\{u,v\})+\{u'v'\}. 
\end{equation*} If $G$ contains an edge $u'v'$, then $G'$ will contain two parallel edges $u'v'$. Observe that $G'$ is a claw-free bridgeless cubic graph with $|V(G')|<|V(G)|=n$, hence by inductive hypothesis it admits a normal edge-coloring $f$ with at most $6$ colors. Assume that in $f$ the new edge $u'v'$ is colored with $1$, and the other $2$ edges incident to $u'$ are colored with $2$ and $3$. Consider an edge-coloring of $G$ obtained from $f$ as follows: color the edges $uu'$ and $vv'$ with $1$, one of parallel edges $uv$ with $2$ and the other edge $uv$ with $3$. It can be easily checked that this new coloring is a normal edge-coloring of $G$ with at most $6$ colors.

In the following, we can assume that $G$ is simple. Now, we apply Proposition \ref{prop:OumClawfreebridgelessCharac}. If $G$ is $K_4$ or a ring of diamonds, then $G$ is $3$-edge-colorable, hence $\chi'_{N}(G)\leq 3$. Thus, without loss of generality, we can assume that there is a connected bridgeless cubic graph $H$, such that $G$ can be obtained from $H$ by replacing some edges of $H$ with a string of diamonds and all vertices of $H$ with a triangle.

Let us show that we can also assume that $G$ contains no diamond. On the opposite assumption, consider a diamond $D$ of $G$. Let $u$ and $v$ be the $2$ vertices of $G$ that have degree $2$ in $D$. Let $u'$ and $v'$ be the neighbours of $u$ and $v$, respectively, that lie outside $D$. Consider a cubic graph $G'$ defined as follows:
\begin{equation*}
    G'=(G-V(D))+\{u'v'\}.
\end{equation*} Observe that $G'$ is a claw-free bridgeless cubic graph with $|V(G')|<|V(G)|$, hence by inductive hypothesis it admits a normal edge-coloring $f$ with at most $6$ colors. We consider two cases. \\

Case 1: The edge $u'v'$ is poor with respect to $f$. Assume that $f(u'v')=3$ and the other neighbours of $u'$ and $v'$ are colored with $1$ and $2$. Consider a coloring of $G$ obtained from $f$ by coloring $uu'$ and $vv'$ with $3$, the spanning $4$ cycle of $D$ with $1$ and $2$, alternatively, and the remaining uncolored edge of $D$ with $3$. It can be easily checked that this new coloring is a normal edge-coloring of $G$ with at most $6$ colors. \\

Case 2: The edge $u'v'$ is rich with respect to $f$. Assume that $f(u'v')=5$, the other edges incident to $u'$ are colored with $1$ and $2$, and the other edges incident to $v'$ are colored with $3$ and $4$. Consider a coloring of $G$ obtained from $f$ by coloring $uu'$ and $vv'$ with $5$, the $2$ edges of the spanning $4$ cycle of $D$ that are incident to $u$ with $1$ and $2$, and the other $2$ edges of the cycle with $3$ and $4$, and finally the remaining uncolored edge of $D$ with $6$. It can be easily checked that this new coloring is a normal edge-coloring of $G$ with at most $6$ colors. \\

Thus, we can assume that $G$ contains no diamond, hence $G$ is obtained from $H$ by replacing every vertex of $H$ by a triangle. Since any bridgeless cubic graph contains a perfect matching, by Lemma \ref{lem:HperfectMatching}, $G$ admits a normal $6$-edge-coloring. The proof of the theorem is complete.
\end{proof}


\subsection{Permutation snarks}
\label{subsec:PermSnarks}

In this section, we introduce cycle permutation cubic graphs and show that they admit a normal $6$-edge-coloring.

A \emph{cycle permutation cubic graph} is a cubic graph of order $2n$ which admits a $2$-factor consisting of two disjoint chordless $n$-cycles $C_1$ and $C_2$.

Permutation graphs were first introduced by Chartrand and Harary in 1967~\cite{CH67}, and cycle permutation graphs were given this name in~\cite{Ri79}, but can also be found in~\cite{Ri84} and other references. 

Let $G=C(n,p)$ be a cycle permutation graph, for some permutation $p \in S_n$, having cycles $C_1$ (external) and $C_2$ (internal). We set the notation $p(i)=p_i$ for all $i \in \{0, \ldots, n-1\}$.
Without loss of generality, we can assume $p_0=0$ and we fix the following labelling on the vertices of $G$, indices taken modulo $n$:

\begin{itemize}
\item the vertices of the cycle $C_1$ are $u_0, \ldots, u_{n-1}$ with $u_iu_{i+1} \in E(G)$;
\item the vertices of the cycle $C_2$ are $v_0, \ldots, v_{n-1}$ with $v_iv_{i+1} \in E(G)$;
\item the edges given by the permutation $p$ are $v_iu_{p_i}$ and they form a perfect matching $M$ of $G$.
\end{itemize}

Along the entire presentation lower indices will be taken modulo $n$. The following well-known fact is an easy consequence of the definition of $C(n,p)$:
\begin{proposition}\label{pro:neven}
Let $C(n,p)$ be a cycle permutation graph with $n$ even. Then, $C(n,p)$ is $3$-edge-colorable.
\end{proposition}
\begin{proof}
Since $C_1$ and $C_2$ are even cycles, they admit a $2$-edge-coloring with colors $1$ and $2$. We obtain a $3$-edge coloring by giving a third color to all edges of the perfect matching $M$. 
\end{proof}

By Proposition \ref{pro:neven} and the fact that a $3$-edge-colorable cubic graph always admits a normal $5$-edge-coloring,  we can only focus on the case that $n$ is odd. In this case, the permutation graph $C(n,p)$ could be not $3$-edge-colorable. Clearly, the Petersen graph is the cycle permutation graph $P=C(5,(1,4,2,3))$ and it is not $3$-edge-colorable.


\begin{lemma}\label{lem:permsnarks}
Let $G=C(n,p)$ be a cycle permutation graph and $n>5$. Then, there exist $u_iv_j, u_hv_k \in E(G)$ such that both $u_iu_h$ and $v_jv_k$ are not edges of $G$.
\end{lemma}
\begin{proof}
Observe that the condition $n>5$ is equivalent to $\binom{n}{2}-n>n$. Now, note that the number of $2$-sets $\{u_i, u_h\}$ is $\binom{n}{2}$ and only $n$ of them are edges of $G$, thus there are $\binom{n}{2}-n$ 2-sets which are not edges. Now, if we look at the corresponding $2$-set $\{v_j, v_k\}$ consisting of the two vertices of $C_2$ adjacent to $u_i$ and $u_h$, then only $n$ of them are edges of $G$. Since $\binom{n}{2}-n>n$, we have that there is a $2$-set $\{u_i, u_h\}$ which is not an edge whose corresponding $2$-set $\{v_j,v_k\}$ does not form an edge in $G$.
\end{proof}

\begin{remark}
Previous lemma cannot be extended to the case $n=5$ due to the Petersen graph.
\end{remark}

The following represents the main result of this section:

\begin{theorem}\label{thm:PermGraphs}
Let $C(n,p)$ be a cycle permutation graph. Then, $\chi'_{N}(C(n,p))\leq 6$. 
\end{theorem}
\begin{proof}
If $n\leq 5$, then $C(n,p)$ has at most $10$ vertices. Hence, it is either $3$-edge-colorable or the Petersen graph: in both cases it admits a normal edge-coloring with at most $5$ colors.

From now on we can assume $n>5$. Then, it follows by Lemma \ref{lem:permsnarks} that the graph $C(n,p)$ has two edges $f=u_iv_j$, $g=u_hv_k$ such that the ends of $f$ are not adjacent to the ends of $g$.

Now, we exhibit a normal $6$-edge-coloring of $C(n,p)$. Firstly, we construct a specific normal $7$-edge-coloring and then we modify it to obtain a normal $6$-edge-coloring.

Denote by $\{001,010,100,011,101,110,111\}$ the nonzero elements of the elementary abelian group $\mathbb{Z}_2^3$. Set $\phi(f)=001$, $\phi(g)=010$ and $\phi(e)=011$ for all other edges $e \in M $. Observe that since we can assume that $n$ is odd (Proposition \ref{pro:neven}), we have: $\sum_{e \in M} \phi(e)=000$. Define $\phi(u_0u_1)=\phi(v_0v_1)=100$, and extend $\phi$ to a nowhere-zero $\mathbb{Z}_2^3$-flow $\phi$ of $C(n,p)$ as it is done in the proof of Lemma 5.2 in \cite{HolySkoJCTB2004}. As it is argued in \cite{HolySkoJCTB2004}, this is possible since  the flow-value of all other edges of $C(n,p)$ is uniquely induced by the values of $\phi$ already assigned and by the fact that the sum of the flows on the edges incident a given vertex must be zero (i.e. $000$).

Every nowhere-zero $\mathbb{Z}_2^3$-flow $\phi$ can be seen as a normal $7$-edge-coloring of $C(n,p)$ (see Theorem 5 in \cite{Normal7flows}). Moreover, the first entry of the flow on edges of $M$ is $0$, and the first entry of the flow on edges outside $M$ is $1$. Hence no edge of $M$ has flow value equal to that of an edge lying outside $M$. Now we slightly modify $\phi$ to obtain a normal $6$-edge-coloring of $C(n,p)$. Let $c$ be the edge-coloring of $C(n,p)$ defined in the following way:
$$c(e)=\phi(e) \text{ for all } e \in E(C(n,p)) \setminus \{f,g\},$$ $$c(g)=c(f)=\phi(f).$$

Observe that $c$ is a $6$-edge-coloring (as it misses the value $010$). Moreover, $c$ is a normal edge-coloring since $\phi$ is a normal edge-coloring and the two edges $f$ and $g$ are not incident to a common edge. The proof is complete.
\end{proof}

\subsection{Treelike snarks}
\label{subsec:TreeLikeSnarks}

In this subsection, we verify Conjecture \ref{conj:6normalBridgelessCubic} in the class of treelike snarks \cite{TreeLike}. First, we start with the necessary definitions. In the subsection, we view each edge of a graph as comprised of two semi-edges. Let $P_0$ be the $5$-zone from Figure \ref{fig:PetersenFragment}, where the loose semi-edges are labeled as $b_1,...,b_5$.


\begin{figure}[ht]
\centering

  	\begin{center}

		\begin{tikzpicture}[scale=0.75]
			
			\node at (6.2, 0.86) {$b_1$};
			\node at (4.7, -1.3) {$b_2$};
			
			\node at (4.35, -0.35) {$b_3$};
			
			\node at (-0.25, -1.4) {$b_4$};
			\node at (-2.25, 0) {$b_5$};
			
			\tikzstyle{every node}=[circle, draw, fill=black!50,
                        inner sep=0pt, minimum width=4pt]
																								
			\node[circle,fill=black,draw] at (0,0) (n00) {};
			\node[circle,fill=black,draw] at (1,0) (n10) {};
			\node[circle,fill=black,draw] at (2,0) (n20) {};

			\node[circle,fill=black,draw] at (0,1) (n01) {};
			\node[circle,fill=black,draw] at (2,1) (n21) {};
			
			\node[circle,fill=black,draw] at (0,2) (n02) {};
			\node[circle,fill=black,draw] at (1,2) (n12) {};
			\node[circle,fill=black,draw] at (2,2) (n22) {};
			
			\node[circle,fill=black,draw] at (3,-1) (n3m1) {};
			\node[circle,fill=black,draw] at (5,1) (n51) {};
			\node[circle,fill=black,draw] at (4,0) (n40) {};

			\draw (0,2) -- (-2,0);
			\draw (2,0) -- (0,-1.5);
			
			\draw (3,-1) -- (4.5,-1.5);
			\draw (5,1) -- (6,0.66);

			\path[every node]
			(n00) edge  (n10)
			edge (n01)
			edge (n3m1)
			
			(n10) edge  (n20)
			edge (n12)
			
			(n01) edge  (n21)
			
			(n02) edge  (n12)
			
			(n12) edge  (n22)
			
			(n01) edge (n02)
			
			(n20) edge (n21)
			(n21) edge (n22)
			(n22) edge (n51)
			(n40) edge (n3m1)
			edge (n51)
			edge [bend left] (3,-2)


			;

		\end{tikzpicture}
																
	\end{center}
	\caption{The $5$-zone $P_{0}$.}\label{fig:PetersenFragment}

\end{figure}
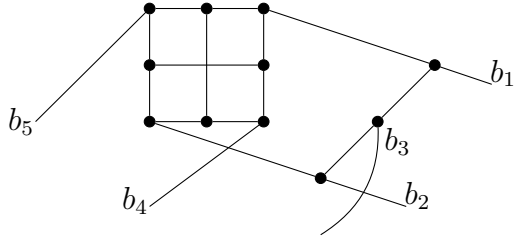

A Halin graph \cite{Halin,Corneujols} is a plane graph that is obtained from a planar representation of a tree without degree-two vertices by joining the leaves of the tree in a cycle. The cycle has as set of its vertices the leaves of the tree. We assume that the leaves of the tree are $l_1,..., l_n$ and this order is the clockwise order of the cycle. Now, let $K$ be a cubic Halin graph with $|V(K)|\geq 4$, and let $T$ and $C$ be the corresponding tree and the cycle of $K$, respectively. The treelike snark $G(T,C)$ \cite{TreeLike} is obtained as follows: take $n$ copies of the 5-zone $P_0$, and identify the copy of the unique end of $b_3$ in the $i$th copy $P_{0}^{i}$ with the leaf $l_i$. Then join $b_{4}^{i}$ to $b_{2}^{i+1}$ and $b_{5}^{i}$ to $b_{1}^{i+1}$ for $i,i+1\in [n]$.

%

Now, we are going to obtain the main result of this subsection:

\begin{theorem}
\label{thm:TreeLike} For any treelike snark $G(T,C)$, we have $\chi'_N(G(T,C))\leq 6$.
\end{theorem}

\begin{proof} Let $G(T,C)$ be a treelike snark, and let $K$ be the corresponding cubic Halin graph that is composed of the tree $T$ and the cycle $C$. By definition any vertex in $T$ is either of degree one or degree three. Let us consider a graph $H$ obtained from $G(T,C)$ as follows: remove the vertices of $T$ that are not leaves, and for all $i\in [n]$ contract the subpath of $P_0^{i}$ with three vertices containing $l_i$ and its two neighbors to a vertex $t_i$. Moreover, contract the remaining eight vertices of $P_0^{i}$ to $z_i$ (see Figure \ref{fig:PetersenFragment}). We keep the parallel edges that arise during the contraction.

It is easy to see that $H$ is a 4-regular graph such that each vertex of $H$ has two neighbors (the underlying simple graph of $H$ is a cycle). Moreover, $H$ contains even number of vertices ($V(H)=\{z_1,t_1,...,z_n,t_n\}$). Since $|V(K)|\geq 4$, we have that $|V(H)|\geq 6$. 

Now, we are going to describe a normal 6-edge-coloring of $G(T,C)$. Let our six colors be $1,...,6$. There is always a non-proper edge coloring $f_H$ of $H$ with colors $1,2,3$, such that the four edges incident to $z_i$ have the same color, and the color of the four edges incident with $z_i$ and the color of the four edges incident with $z_{i+1}$ are different for all $i,i+1\in [n]$.

Since $\{z_i|i\in [n]\}$ is an independent set of $H$, such coloring always exists. For $n\geq 3$ (i.e. $|V(H)|\geq 6$), all three colors 1,2,3 could be used. Then the multi edges $z_it_i$ and $t_iz_{i+1}$ have different colors. Without loss of generality, assume that $t_n$ is incident with edges of colors 1
and 2, and $t_1$ is incident with edges of colors 1 and 3. Below, we assume that $a,b,c$ is a permutation of $1,2,3$ (that is, $\{a,b,c\}=\{1,2,3\}$).

%
%

We extend $f_H$ to a 6-edge-coloring of $G(T,C)$ as follows. First, since the maximum degree in $T$ is three, we find a proper 3-edge-coloring $f_T$ of $T$ with colors $4,5,6$. Now, we are going to obtain the coloring around the eight vertices corresponding $z_j$ if we had the coloring around the eight vertices corresponding to $z_{j-1}$ (Figure \ref{fig:CaseA1}). We split the proof in two cases and four subcases: we assume that $\alpha, \beta, \gamma$ is a permutation of $4,5,6$ (that is, $\{\alpha, \beta, \gamma\}=\{4,5,6\}$).\\

Case A: Assume that the four edges incident to $z_j$ are colored with $a$, and the edges $z_{j-1}t_{j-1}$ and $t_jz_{j+1}$ are colored
with the same color, say $b$. We differ two subcases A1 and A2 depending whether the pendant edges of $T$ corresponding to $t_{j-1}$ and $t_j$ have the same color or not. When these colors are the same (subcase A1, these edges are of color $\alpha$), the edge-coloring is described on Figure \ref{fig:CaseA1}. When these colors are different (Subcase A2, these edges are of colors $\alpha$ and $\beta$, respectively), the corresponding edge-coloring is described on Figure \ref{fig:CaseA2}.\\

\begin{figure}[ht]
\centering
\begin{minipage}[b]{.5\textwidth}
  \begin{center}

\begin{tikzpicture}[scale=0.75]
			
			
			\node at (-1.5, 1.7) {$a$};
			\node at (-3.75, 1.1) {$b$};
			\node at (-2.3, 0.7) {$\beta$};
			\node at (-1.3, -0.3) {$\gamma$};
			\node at (-2.3, -0.3) {$\alpha$};
			
			\node at (3.5, 1.7) {$a$};
			\node at (5.5, 1.1) {$b$};
			\node at (4.3, 0.7) {$\beta$};
			\node at (3.3, -0.3) {$\gamma$};
			\node at (4.3, -0.3) {$\alpha$};
			
			\node at (-0.35, -0.95) {$a$};
			\node at (-2.7, -1.6) {$b$};
			
			\node at (2.15, -0.95) {$a$};
			\node at (4.6, -1.6) {$b$};
			
			
			\node at (0.5, 0.3) {$b$};
			\node at (1.5, 0.2) {$c$};
			
			
		    \node at (0.5, 2.3) {$b$};
			\node at (1.5, 2.2) {$c$};
			
			
		    \node at (-0.3, 1.5) {$\beta$};
			\node at (-0.3, 0.5) {$\gamma$};
			
			
		    \node at (2.3, 1.5) {$\gamma$};
			\node at (2.3, 0.5) {$\beta$};
			
			
		    \node at (1.3, 1.5) {$\alpha$};
			\node at (0.5, 1.2) {$\alpha$};

			\tikzstyle{every node}=[circle, draw, fill=black!50,
                        inner sep=0pt, minimum width=4pt]
																								
			\node[circle,fill=black,draw] at (0,0) (n00) {};
			\node[circle,fill=black,draw] at (1,0) (n10) {};
			\node[circle,fill=black,draw] at (2,0) (n20) {};

			\node[circle,fill=black,draw] at (0,1) (n01) {};
			\node[circle,fill=black,draw] at (2,1) (n21) {};
			
			\node[circle,fill=black,draw] at (0,2) (n02) {};
			\node[circle,fill=black,draw] at (1,2) (n12) {};
			\node[circle,fill=black,draw] at (2,2) (n22) {};
			
			\node[circle,fill=black,draw] at (3,-1) (n3m1) {};
			\node[circle,fill=black,draw] at (5,1) (n51) {};
			\node[circle,fill=black,draw] at (4,0) (n40) {};
			
			\node[circle,fill=black,draw] at (-1,-1) (nm1m1) {};
			\node[circle,fill=black,draw] at (-3,1) (nm31) {};
			\node[circle,fill=black,draw] at (-2,0) (nm20) {};

			
			\draw (3,-1) -- (4.5,-1.5);
			\draw (5,1) -- (6,0.66);
			
			\draw (-1,-1) -- (-2.5,-1.5);
			\draw (-3,1) -- (-4,0.66);

			\path[every node]
			(n00) edge  (n10)
			edge (n01)
			edge (n3m1)
			
			(n10) edge  (n20)
			edge (n12)
			
			(n01) edge  (n21)
			
			(n02) edge  (n12)
			edge (nm31)
			
			(n12) edge  (n22)
			
			(n01) edge (n02)
			
			(n20) edge (n21)
			edge (nm1m1)
			(n21) edge (n22)
			(n22) edge (n51)
			(n40) edge (n3m1)
			edge (n51)
			edge [bend left] (3,-2)
			
			(nm20) edge (nm1m1)
			edge (nm31)
			edge [bend right] (-1,-2)


			;

		\end{tikzpicture}
																
	\end{center}
	
	\caption{The case A1.}\label{fig:CaseA1}
\end{minipage}%
\begin{minipage}[b]{.5\textwidth}
  	\begin{center}

\begin{tikzpicture}[scale=0.75]
			
			
			\node at (-1.5, 1.7) {$a$};
			\node at (-3.75, 1.1) {$b$};
			\node at (-2.3, 0.7) {$\beta$};
			\node at (-1.3, -0.3) {$\gamma$};
			\node at (-2.3, -0.3) {$\alpha$};
			
			\node at (3.5, 1.7) {$a$};
			\node at (5.5, 1.1) {$b$};
			\node at (4.3, 0.7) {$\gamma$};
			\node at (3.3, -0.3) {$\alpha$};
			\node at (4.3, -0.3) {$\beta$};
			
			\node at (-0.35, -0.95) {$a$};
			\node at (-2.7, -1.6) {$b$};
			
			\node at (2.15, -0.95) {$a$};
			\node at (4.6, -1.6) {$b$};
			
			
			\node at (0.5, 0.3) {$b$};
			\node at (1.5, 0.2) {$c$};
			
			
		    \node at (0.5, 2.3) {$c$};
			\node at (1.5, 2.2) {$b$};
			
			
		    \node at (-0.3, 1.5) {$\gamma$};
			\node at (-0.3, 0.5) {$\alpha$};
			
			
		    \node at (2.3, 1.5) {$\gamma$};
			\node at (2.3, 0.5) {$\alpha$};
			
			
		    \node at (1.3, 1.5) {$\beta$};
			\node at (0.5, 1.2) {$\beta$};
			
			\tikzstyle{every node}=[circle, draw, fill=black!50,
                        inner sep=0pt, minimum width=4pt]
																								
			\node[circle,fill=black,draw] at (0,0) (n00) {};
			\node[circle,fill=black,draw] at (1,0) (n10) {};
			\node[circle,fill=black,draw] at (2,0) (n20) {};

			\node[circle,fill=black,draw] at (0,1) (n01) {};
			\node[circle,fill=black,draw] at (2,1) (n21) {};
			
			\node[circle,fill=black,draw] at (0,2) (n02) {};
			\node[circle,fill=black,draw] at (1,2) (n12) {};
			\node[circle,fill=black,draw] at (2,2) (n22) {};
			
			\node[circle,fill=black,draw] at (3,-1) (n3m1) {};
			\node[circle,fill=black,draw] at (5,1) (n51) {};
			\node[circle,fill=black,draw] at (4,0) (n40) {};
			
			\node[circle,fill=black,draw] at (-1,-1) (nm1m1) {};
			\node[circle,fill=black,draw] at (-3,1) (nm31) {};
			\node[circle,fill=black,draw] at (-2,0) (nm20) {};

			
			\draw (3,-1) -- (4.5,-1.5);
			\draw (5,1) -- (6,0.66);
			
			\draw (-1,-1) -- (-2.5,-1.5);
			\draw (-3,1) -- (-4,0.66);

			\path[every node]
			(n00) edge  (n10)
			edge (n01)
			edge (n3m1)
			
			(n10) edge  (n20)
			edge (n12)
			
			(n01) edge  (n21)
			
			(n02) edge  (n12)
			edge (nm31)
			
			(n12) edge  (n22)
			
			(n01) edge (n02)
			
			(n20) edge (n21)
			edge (nm1m1)
			(n21) edge (n22)
			(n22) edge (n51)
			(n40) edge (n3m1)
			edge (n51)
			edge [bend left] (3,-2)
			
			(nm20) edge (nm1m1)
			edge (nm31)
			edge [bend right] (-1,-2)


			;

		\end{tikzpicture}
																
	\end{center}
	\caption{The case A2.}\label{fig:CaseA2}
\end{minipage}
\end{figure}

Case B: Assume that the four edges incident to $z_j$ are colored with $a$, and the edges $z_{j-1}t_{j-1}$ and $t_jz_{j+1}$ are colored with different colors, say $b$ and $c$, respectively. We differ two subcases B1 and B2 depending whether the pendant edges of $T$ corresponding to $t_{j-1}$ and $t_j$ have the same color or not. When these colors are the same (subcase B1, these edges are of color $\alpha$), the edge-coloring is described on Figure \ref{fig:CaseB1}. When these colors are different (Subcase B2, these edges are of colors $\alpha$ and $\beta$, respectively), the corresponding edge-coloring is described on Figure \ref{fig:CaseB2}.\\

\begin{figure}[ht]
\centering
\begin{minipage}[b]{.5\textwidth}
  \begin{center}

\begin{tikzpicture}[scale=0.75]
			
			
			\node at (-1.5, 1.7) {$a$};
			\node at (-3.75, 1.1) {$b$};
			\node at (-2.3, 0.7) {$\beta$};
			\node at (-1.3, -0.3) {$\gamma$};
			\node at (-2.3, -0.3) {$\alpha$};
			
			\node at (3.5, 1.7) {$a$};
			\node at (5.5, 1.1) {$c$};
			\node at (4.3, 0.7) {$\beta$};
			\node at (3.3, -0.3) {$\gamma$};
			\node at (4.3, -0.3) {$\alpha$};
			
			\node at (-0.35, -0.95) {$a$};
			\node at (-2.7, -1.6) {$b$};
			
			\node at (2.15, -0.95) {$a$};
			\node at (4.6, -1.6) {$c$};
			
			
			\node at (0.5, 0.3) {$b$};
			\node at (1.5, 0.2) {$c$};
			
			
		    \node at (0.5, 2.3) {$c$};
			\node at (1.5, 2.2) {$b$};
			
			
		    \node at (-0.3, 1.5) {$\alpha$};
			\node at (-0.3, 0.5) {$\beta$};
			
			
		    \node at (2.3, 1.5) {$\alpha$};
			\node at (2.3, 0.5) {$\beta$};
			
			
		    \node at (1.3, 1.5) {$\gamma$};
			\node at (0.5, 1.2) {$\gamma$};
			
			\tikzstyle{every node}=[circle, draw, fill=black!50,
                        inner sep=0pt, minimum width=4pt]
																								
			\node[circle,fill=black,draw] at (0,0) (n00) {};
			\node[circle,fill=black,draw] at (1,0) (n10) {};
			\node[circle,fill=black,draw] at (2,0) (n20) {};

			\node[circle,fill=black,draw] at (0,1) (n01) {};
			\node[circle,fill=black,draw] at (2,1) (n21) {};
			
			\node[circle,fill=black,draw] at (0,2) (n02) {};
			\node[circle,fill=black,draw] at (1,2) (n12) {};
			\node[circle,fill=black,draw] at (2,2) (n22) {};
			
			\node[circle,fill=black,draw] at (3,-1) (n3m1) {};
			\node[circle,fill=black,draw] at (5,1) (n51) {};
			\node[circle,fill=black,draw] at (4,0) (n40) {};
			
			\node[circle,fill=black,draw] at (-1,-1) (nm1m1) {};
			\node[circle,fill=black,draw] at (-3,1) (nm31) {};
			\node[circle,fill=black,draw] at (-2,0) (nm20) {};

			
			\draw (3,-1) -- (4.5,-1.5);
			\draw (5,1) -- (6,0.66);
			
			\draw (-1,-1) -- (-2.5,-1.5);
			\draw (-3,1) -- (-4,0.66);

			\path[every node]
			(n00) edge  (n10)
			edge (n01)
			edge (n3m1)
			
			(n10) edge  (n20)
			edge (n12)
			
			(n01) edge  (n21)
			
			(n02) edge  (n12)
			edge (nm31)
			
			(n12) edge  (n22)
			
			(n01) edge (n02)
			
			(n20) edge (n21)
			edge (nm1m1)
			(n21) edge (n22)
			(n22) edge (n51)
			(n40) edge (n3m1)
			edge (n51)
			edge [bend left] (3,-2)
			
			(nm20) edge (nm1m1)
			edge (nm31)
			edge [bend right] (-1,-2)


			;

		\end{tikzpicture}
																
	\end{center}
	
	\caption{The case B1.}\label{fig:CaseB1}
\end{minipage}%
\begin{minipage}[b]{.5\textwidth}
  	\begin{center}

\begin{tikzpicture}[scale=0.75]
			
			
			\node at (-1.5, 1.7) {$a$};
			\node at (-3.75, 1.1) {$b$};
			\node at (-2.3, 0.7) {$\beta$};
			\node at (-1.3, -0.3) {$\gamma$};
			\node at (-2.3, -0.3) {$\alpha$};
			
			\node at (3.5, 1.7) {$a$};
			\node at (5.5, 1.1) {$c$};
			\node at (4.3, 0.7) {$\alpha$};
			\node at (3.3, -0.3) {$\gamma$};
			\node at (4.3, -0.3) {$\beta$};
			
			\node at (-0.35, -0.95) {$a$};
			\node at (-2.7, -1.6) {$b$};
			
			\node at (2.15, -0.95) {$a$};
			\node at (4.6, -1.6) {$c$};
			
			
			\node at (0.5, 0.3) {$b$};
			\node at (1.5, 0.2) {$c$};
			
			
		    \node at (0.5, 2.3) {$c$};
			\node at (1.5, 2.2) {$b$};
			
			
		    \node at (-0.3, 1.5) {$\alpha$};
			\node at (-0.3, 0.5) {$\beta$};
			
			
		    \node at (2.3, 1.5) {$\beta$};
			\node at (2.3, 0.5) {$\alpha$};
			
			
		    \node at (1.3, 1.5) {$\gamma$};
			\node at (0.5, 1.2) {$\gamma$};
			
			\tikzstyle{every node}=[circle, draw, fill=black!50,
                        inner sep=0pt, minimum width=4pt]
																								
			\node[circle,fill=black,draw] at (0,0) (n00) {};
			\node[circle,fill=black,draw] at (1,0) (n10) {};
			\node[circle,fill=black,draw] at (2,0) (n20) {};

			\node[circle,fill=black,draw] at (0,1) (n01) {};
			\node[circle,fill=black,draw] at (2,1) (n21) {};
			
			\node[circle,fill=black,draw] at (0,2) (n02) {};
			\node[circle,fill=black,draw] at (1,2) (n12) {};
			\node[circle,fill=black,draw] at (2,2) (n22) {};
			
			\node[circle,fill=black,draw] at (3,-1) (n3m1) {};
			\node[circle,fill=black,draw] at (5,1) (n51) {};
			\node[circle,fill=black,draw] at (4,0) (n40) {};
			
			\node[circle,fill=black,draw] at (-1,-1) (nm1m1) {};
			\node[circle,fill=black,draw] at (-3,1) (nm31) {};
			\node[circle,fill=black,draw] at (-2,0) (nm20) {};

			
			\draw (3,-1) -- (4.5,-1.5);
			\draw (5,1) -- (6,0.66);
			
			\draw (-1,-1) -- (-2.5,-1.5);
			\draw (-3,1) -- (-4,0.66);

			\path[every node]
			(n00) edge  (n10)
			edge (n01)
			edge (n3m1)
			
			(n10) edge  (n20)
			edge (n12)
			
			(n01) edge  (n21)
			
			(n02) edge  (n12)
			edge (nm31)
			
			(n12) edge  (n22)
			
			(n01) edge (n02)
			
			(n20) edge (n21)
			edge (nm1m1)
			(n21) edge (n22)
			(n22) edge (n51)
			(n40) edge (n3m1)
			edge (n51)
			edge [bend left] (3,-2)
			
			(nm20) edge (nm1m1)
			edge (nm31)
			edge [bend right] (-1,-2)


			;

		\end{tikzpicture}
																
	\end{center}
	\caption{The case B2.}\label{fig:CaseB2}
\end{minipage}
\end{figure}

Now, in order to inductively color $G(T,C)$, we consider the leave $l_n$ of $T$ corresponding to $t_n$. Fix one of the two possible ways of coloring the two edges incident $l_n$ not in $T$ with the two available colors in $\{1,2,3\}$. Then, we extend the coloring to the whole graph $G(T,C)$ by considering the blocks corresponding to $z_1,..., z_n$ in this order, and coloring the uncolored edges of $G(T,C)$ according to cases A and B. The edges of 3-paths corresponding to $t_j$s are colored by two colors of the two edges of $T$ adjacent to it. This will result to an edge-coloring of the whole graph. 

Let us consider the leave of $T$ corresponding to $t_n$. If during the consideration of $z_n$ (by applying the cases A or B), we did not flip the colors of the subpath of length two corresponding to $t_n$ (the order of colors that we fixed initially), then we stop. Otherwise, if the colors of these two edges have been flipped, then we differ two subcases C1 and C2 depending whether the pendant edges of $T$ corresponding to $t_{n}$ and $t_1$ have the same color or not. When these colors are the same (subcase C1, these edges are of color $\alpha$), we recolor the block corresponding to $z_1$ as it is described on Figure \ref{fig:CaseC1}. When these colors are different (Subcase C2, these edges are of colors $\alpha$ and $\beta$, respectively), we recolor the block corresponding to $z_1$ as it is described on Figure \ref{fig:CaseC2}. It is matter of direct verification that the resulting 6-edge-coloring is a normal edge-coloring of $G(T,C)$. The proof is complete.

\begin{figure}[ht]
\centering
\begin{minipage}[b]{.5\textwidth}
  \begin{center}

\begin{tikzpicture}[scale=0.75]
			
			
			\node at (-1.5, 1.7) {$a$};
			\node at (-3.75, 1.1) {$b$};
			\node at (-2.3, 0.7) {$\gamma$};
			\node at (-1.3, -0.3) {$\beta$};
			\node at (-2.3, -0.3) {$\alpha$};
			
			\node at (3.5, 1.7) {$a$};
			\node at (5.5, 1.1) {$c$};
			\node at (4.3, 0.7) {$\beta$};
			\node at (3.3, -0.3) {$\gamma$};
			\node at (4.3, -0.3) {$\alpha$};
			
			\node at (-0.35, -0.95) {$a$};
			\node at (-2.7, -1.6) {$b$};
			
			\node at (2.15, -0.95) {$a$};
			\node at (4.6, -1.6) {$c$};
			
			
			\node at (0.5, 0.3) {$c$};
			\node at (1.5, 0.3) {$b$};
			
			
		    \node at (0.5, 2.3) {$c$};
			\node at (1.5, 2.3) {$b$};
			
			
		    \node at (-0.3, 1.5) {$\beta$};
			\node at (-0.3, 0.5) {$\gamma$};
			
			
		    \node at (2.3, 1.5) {$\gamma$};
			\node at (2.3, 0.5) {$\beta$};
			
			
		    \node at (1.3, 1.5) {$\alpha$};
			\node at (0.5, 1.2) {$\alpha$};
			
			\tikzstyle{every node}=[circle, draw, fill=black!50,
                        inner sep=0pt, minimum width=4pt]
																								
			\node[circle,fill=black,draw] at (0,0) (n00) {};
			\node[circle,fill=black,draw] at (1,0) (n10) {};
			\node[circle,fill=black,draw] at (2,0) (n20) {};

			\node[circle,fill=black,draw] at (0,1) (n01) {};
			\node[circle,fill=black,draw] at (2,1) (n21) {};
			
			\node[circle,fill=black,draw] at (0,2) (n02) {};
			\node[circle,fill=black,draw] at (1,2) (n12) {};
			\node[circle,fill=black,draw] at (2,2) (n22) {};
			
			\node[circle,fill=black,draw] at (3,-1) (n3m1) {};
			\node[circle,fill=black,draw] at (5,1) (n51) {};
			\node[circle,fill=black,draw] at (4,0) (n40) {};
			
			\node[circle,fill=black,draw] at (-1,-1) (nm1m1) {};
			\node[circle,fill=black,draw] at (-3,1) (nm31) {};
			\node[circle,fill=black,draw] at (-2,0) (nm20) {};

			
			\draw (3,-1) -- (4.5,-1.5);
			\draw (5,1) -- (6,0.66);
			
			\draw (-1,-1) -- (-2.5,-1.5);
			\draw (-3,1) -- (-4,0.66);

			\path[every node]
			(n00) edge  (n10)
			edge (n01)
			edge (n3m1)
			
			(n10) edge  (n20)
			edge (n12)
			
			(n01) edge  (n21)
			
			(n02) edge  (n12)
			edge (nm31)
			
			(n12) edge  (n22)
			
			(n01) edge (n02)
			
			(n20) edge (n21)
			edge (nm1m1)
			(n21) edge (n22)
			(n22) edge (n51)
			(n40) edge (n3m1)
			edge (n51)
			edge [bend left] (3,-2)
			
			(nm20) edge (nm1m1)
			edge (nm31)
			edge [bend right] (-1,-2)


			;

		\end{tikzpicture}
																
	\end{center}
	
	\caption{The case C1.}\label{fig:CaseC1}
\end{minipage}%
\begin{minipage}[b]{.5\textwidth}
  	\begin{center}

\begin{tikzpicture}[scale=0.75]
			
			
			\node at (-1.5, 1.7) {$a$};
			\node at (-3.75, 1.1) {$b$};
			\node at (-2.3, 0.7) {$\gamma$};
			\node at (-1.3, -0.3) {$\beta$};
			\node at (-2.3, -0.3) {$\alpha$};
			
			\node at (3.5, 1.7) {$a$};
			\node at (5.5, 1.1) {$c$};
			\node at (4.3, 0.7) {$\alpha$};
			\node at (3.3, -0.3) {$\gamma$};
			\node at (4.3, -0.3) {$\beta$};
			
			\node at (-0.35, -0.95) {$a$};
			\node at (-2.7, -1.6) {$b$};
			
			\node at (2.15, -0.95) {$a$};
			\node at (4.6, -1.6) {$c$};
			
			
			\node at (0.5, 0.3) {$c$};
			\node at (1.5, 0.3) {$b$};
			
			
		    \node at (0.5, 2.3) {$c$};
			\node at (1.5, 2.3) {$b$};
			
			
		    \node at (-0.3, 1.5) {$\beta$};
			\node at (-0.3, 0.5) {$\gamma$};
			
			
		    \node at (2.3, 1.5) {$\gamma$};
			\node at (2.3, 0.5) {$\beta$};
			
			
		    \node at (1.3, 1.5) {$\alpha$};
			\node at (0.5, 1.2) {$\alpha$};
			
			\tikzstyle{every node}=[circle, draw, fill=black!50,
                        inner sep=0pt, minimum width=4pt]
																								
			\node[circle,fill=black,draw] at (0,0) (n00) {};
			\node[circle,fill=black,draw] at (1,0) (n10) {};
			\node[circle,fill=black,draw] at (2,0) (n20) {};

			\node[circle,fill=black,draw] at (0,1) (n01) {};
			\node[circle,fill=black,draw] at (2,1) (n21) {};
			
			\node[circle,fill=black,draw] at (0,2) (n02) {};
			\node[circle,fill=black,draw] at (1,2) (n12) {};
			\node[circle,fill=black,draw] at (2,2) (n22) {};
			
			\node[circle,fill=black,draw] at (3,-1) (n3m1) {};
			\node[circle,fill=black,draw] at (5,1) (n51) {};
			\node[circle,fill=black,draw] at (4,0) (n40) {};
			
			\node[circle,fill=black,draw] at (-1,-1) (nm1m1) {};
			\node[circle,fill=black,draw] at (-3,1) (nm31) {};
			\node[circle,fill=black,draw] at (-2,0) (nm20) {};

			
			\draw (3,-1) -- (4.5,-1.5);
			\draw (5,1) -- (6,0.66);
			
			\draw (-1,-1) -- (-2.5,-1.5);
			\draw (-3,1) -- (-4,0.66);

			\path[every node]
			(n00) edge  (n10)
			edge (n01)
			edge (n3m1)
			
			(n10) edge  (n20)
			edge (n12)
			
			(n01) edge  (n21)
			
			(n02) edge  (n12)
			edge (nm31)
			
			(n12) edge  (n22)
			
			(n01) edge (n02)
			
			(n20) edge (n21)
			edge (nm1m1)
			(n21) edge (n22)
			(n22) edge (n51)
			(n40) edge (n3m1)
			edge (n51)
			edge [bend left] (3,-2)
			
			(nm20) edge (nm1m1)
			edge (nm31)
			edge [bend right] (-1,-2)


			;

		\end{tikzpicture}
																
	\end{center}
	\caption{The case C2.}\label{fig:CaseC2}
\end{minipage}
\end{figure}

\end{proof}

\subsection{The number of normal edges in 6-edge-colorings}
\label{subsec:NormalEdges}

If a cubic graph $G$ is edge-colored with $k$ colors, then some edges of $G$ are poor (in the coloring), others are rich (in the coloring) and the rest of edges are neither rich nor poor. Let us say that an edge is normal with respect to an edge-coloring, if it is poor or rich in this coloring. If $f$ is an edge-coloring of $G$, then let $E(f)$ be the set of normal edges of $G$ in $f$. Conjecture \ref{conj:5NormalConj} predicts that any bridgeless cubic graph $G$ admits a 5-edge-coloring $f$, such that $E(f)=E(G)$. In \cite{SamalElecNotes} a result towards this conjecture is obtained which states that any bridgeless cubic graph admits a 5-edge-coloring $f$, such that $|E(f)|\geq \frac{|E(G)|}{3}$. Furthermore, a similar result is recently proved in \cite{PSS} also for 4-edge-colorings. In the light of Conjecture \ref{conj:6normalBridgelessCubic}, one may try to obtain a lower bound for $|E(f)|$, when $f$ is a 6-edge-coloring. Our next theorem addresses this issue.

\begin{theorem}
\label{thm:59Bound} Let $G$ be a bridgeless cubic graph. Then $G$ admits a 6-edge-coloring $f$, such that $|E(f)|\geq \frac{7}{9}\cdot  |E(G)|$.
\end{theorem}

\begin{proof} Let $G$ be a counterexample to the theorem minimizing $|V(G)|$. Clearly, $G$ is connected. Let us show that it has no 2-edge-cuts. Assume that $C=\{e_1, e_2\}$ is a 2-edge-cut. Let $G_1$ and $G_2$ be the two smaller bridgeless cubic graphs arising from the two components of $G-C$ by adding one edge connecting the two degree-two vertices in the same component. We let $h_1$ and $h_2$ be the two added edges of these two graphs, respectively. Since the graphs $G_1$ and $G_2$ are smaller, we have that they admit 6-edge-colorings $f_1$ and $f_2$ such that $|E(f_j)|\geq \frac{7}{9}\cdot |E(G_j)|$, $j=1,2$. By renaming the colors in $G_2$, we can always assume that the colors of $h_1$ and $h_2$ are the same, moreover, the colors appearing in the ends of $e_1$ are also the same. Now, if we color $e_1$ and $e_2$ with the color of $h_1$, then we will have that $e_1$ is always, poor, moreover if at least one of $h_1$ and $h_2$ is normal, then $e_2$ will also be normal. This means that in the resulting coloring $f$ of $G$, we will have:
\[|E(f)|\geq |E(f_1)|+|E(f_2)|.\]
Since, $|E(G)|=|E(G_1)|+|E(G_2)|$, we have
\[\frac{|E(f)|}{|E(G)|}\geq \frac{|E(f_1)|+|E(f_2)|}{|E(G_1)|+|E(G_2)|}\geq \min \left \{ \frac{|E(f_1)|}{|E(G_1)|}, \frac{|E(f_2)|}{|E(G_2)|}\right \}\geq \frac{7}{9}.\]

Thus, $G$ must be 3-connected. Let $M$ be a perfect matching of $G$ that intersects each 3-edge-cut of $G$ in a single edge (Theorem \ref{thm:KaiserKral}). If $\overline{M}$ is the complementary 2-factor of $M$, then $G/\overline{M}$ is 4-edge-connected. Hence it admits a nowhere zero $\mathbb{Z}_2^2$-flow $\theta$ (Theorem \ref{thm:Jaeger4flow}). Let us extend $\theta$ to a nowhere zero $\mathbb{Z}_2^3$-flow $\mu$ of the whole graph $G$ as it is done in the proof of Lemma 5.2 in \cite{HolySkoJCTB2004}: first for any edge $h\in M$, we define the triple $\mu(h)$ as follows: $\mu(h)=(0,\theta(h))$. Now, let $C$ be any cycle of $\overline{M}$. Let $x_0$ be any element of $\mathbb{Z}_2^3$, whose first coordinate is $1$. Assign $x_0$ to an edge of $C$. Then observe that the rest of the values of edges of $C$ are defined uniquely in $\mu$. Moreover, the first coordinates of the values of $\mu$ on $C$ are $1$. Hence for any edges $h_1\in M$ and $h_2\in \overline{M}$, we have $\mu(h_1)\neq \mu(h_2)$. Also observe that for different cycles of $\overline{M}$ we can choose $x_0$ differently. 

Clearly, for some nonzero element $\beta$ of $\mathbb{Z}_2^2$, we have 
\[|\theta^{-1}(\beta)|\leq \frac{1}{3}\cdot |E(G/\overline{M})|=\frac{1}{3}\cdot \frac{|V(G)|}{2}=\frac{|V(G)|}{6}.\]

Denote by $\alpha,\beta,\gamma$ the three nonzero elements of $\mathbb{Z}_2^2$. Denote by $n_\alpha$ (or $n_\gamma$) the number of edges of $\overline{M}$ having an end incident with an edge of $M$ with value $\beta$ and the other end incident with an edge of $M$ with value $\alpha$ (or $\gamma$). The relation $n_\alpha+n_\gamma \leq 4 \cdot \frac{|V(G)|}{6}$ holds, since the second term is larger than the total number of edges which have an end incident with an edge with value $\beta$. Hence, at least one between $n_\alpha$ and $n_\gamma$ is less or equal to $2 \cdot \frac{|V(G)|}{6}$, say $n_\alpha$. 
Consider a mapping $f$ obtained from $\mu$ as follows: if $e\notin \theta^{-1}(\beta)$, then $f(e)=\mu(e)$, otherwise, $f(e)=(0,\alpha)$. Let us show that $f$ is a proper 6-edge-coloring of $G$. First, note that the values of $\mu$ are the seven nonzero elements of $\mathbb{Z}_2^3$. Since, by definition, $f$ does take the value $(0,\beta)$, we have that $f$ takes at most six values. Moreover, since $\mu$ is a nowhere zero $\mathbb{Z}_2^3$-flow, it is a proper coloring. Now, if we look at the edges of $f^{-1}((0, \alpha))$, it is a subset of $M$, hence it is a matching. Thus, taking into account that $\mu$ is a proper coloring and $f^{-1}((0, \alpha))$ is a matching, we have that $f$ is a proper 6-edge-coloring.

In order to complete the proof, let us show that there are at most $n_\alpha$ edges that are neither rich nor poor in $f$. Since $\mu$ is a nowhere zero $\mathbb{Z}_2^3$-flow, it is a normal 7-edge-coloring, hence all edges of $G$ are either poor or rich in $\mu$. In order to construct $f$, we changed the values of $\mu$ on edges of $\mu^{-1}((0, \beta))$. Thus, all edges of $M$ will remain poor or rich in $f$. Moreover, the edges of $\overline{M}$ that are incident to two edges of $M$ with the same value of $\mu$, will remain poor or rich in $f$. The only possibility, when an edge that is neither poor nor rich in $f$ may arise is that when it is adjacent to an edge of $M$ with $\mu$-value $\alpha$ and an edge of $M$ with $\mu$-value $\beta$. But the number of such edges is $n_{\alpha}$. Thus, we may have at most $n_\alpha$ edges that are neither rich nor poor in $f$.
By the choice of $n_{\alpha}$, we have:
\[ |\overline{E(f)}| \leq n_\alpha \leq 2\cdot \frac{|V(G)|}{6}=\frac{|V(G)|}{3}=\frac{2|E(G)|}{9}.\]
Thus, for the resulting 6-edge-coloring $f$, we will have
\[|E(f)|\geq \frac{7|E(G)|}{9}.\]
The proof is complete.
\end{proof}


\section{Future work}
\label{sec:FutureWork}

The main result of \cite{Normal7flows} states that any simple cubic graph admits a normal 7-edge-coloring. There it is also shown that any bridgeless cubic graph $G$ admits a normal 7-edge-coloring (see also \cite{Bilkova12}), and this result is obtained simply by considering a nowhere zero $\mathbb{Z}_2^3$-flow of $G$. One may wonder whether we can choose the nowhere zero $\mathbb{Z}_2^3$-flow $\theta$ of $G$, such that for one nonzero element $\gamma\in \mathbb{Z}_2^3$, we have $\theta^{-1}(\gamma)=\emptyset$. Observe that if such a flow existed in $G$, it would have been a normal 6-edge-coloring of $G$. The next theorem shows that not all bridgeless cubic graphs can have such a nowhere zero $\mathbb{Z}_2^3$-flow.

\begin{theorem}
\label{thm:MissingValueNZ8flow} Let $G$ be a bridgeless cubic graph, and assume that $G$ admits a nowhere zero $\mathbb{Z}_2^3$-flow $f$, such that there is a nonzero $\gamma\in \mathbb{Z}_2^3$, such that $f^{-1}(\gamma)=\emptyset$. Then $G$ is $3$-edge-colorable.
\end{theorem}

\begin{proof} Let $G$ be a bridgeless cubic graph, and let $f$ be a nowhere zero $\mathbb{Z}_2^3$-flow such that $f^{-1}(\gamma)=\emptyset$. By choosing a suitable automorphism of $\mathbb{Z}_2^3$, we can always assume that $\gamma=111$. Observe that the other six nonzero elements of $\mathbb{Z}_2^3$ can be partitioned into three subsets $\Gamma_1$, $\Gamma_2$, $\Gamma_3$ of cardinality two, such that the sum of elements in each $\Gamma_i$ in the group $Z_{2}^{3}$ is equal to $\gamma$.

Let us show that $f^{-1}(\Gamma_j)$ is a matching in $G$. Assume it contains two adjacent edges $x$ and $y$. Then the value of the flow $f$ on the third edge must be $f(x)+f(y)=\gamma$, which contradicts the fact that $f^{-1}(\gamma)=\emptyset$. 

Thus, $f^{-1}(\Gamma_j)$ is a matching for $j=1,2,3$, and clearly these three matchings form a partition of $E(G)$. Hence, $G$ is $3$-edge-colorable. The proof is complete.
\end{proof}

Theorem \ref{thm:MissingValueNZ8flow} implies that, for non-3-edge-colorable cubic graphs, there is no hope to prove Conjecture \ref{conj:6normalBridgelessCubic} with the approach outlined above.

The next approach for proving Conjecture \ref{conj:6normalBridgelessCubic} prompt the proofs of Theorems \ref{thm:PermGraphs} and \ref{thm:59Bound}. It is easy to see that the smallest counterexample to Conjecture \ref{conj:6normalBridgelessCubic} must be a 3-edge-connected graph. Hence, it will follow from the following:

\begin{conjecture}
\label{conj:NonConflictingFlow6} Let $G$ be a 3-edge-connected cubic graph different from the Petersen graph (Figure \ref{fig:Petersen10}). Then $G$ admits a nowhere zero $\mathbb{Z}_2^3$-flow $f$, such that there are two elements $\alpha, \beta \in \mathbb{Z}_2^3$ with
\begin{enumerate}
    \item [(1)] $f^{-1}(\{\alpha, \beta \})$ is a matching in $G$,
    
    \item [(2)] there is no edge $e=uv$ of $G$, such that $u$ is incident to an edge $e_u$ and $v$ is incident to an edge $e_v$ with $f(e_u)=\alpha$ and $f(e_v)=\beta$.
\end{enumerate}
\end{conjecture}

In other words the second condition in Conjecture \ref{conj:NonConflictingFlow6} says that the subgraph induced by the edges in $f^{-1}(\{\alpha, \beta \})$ is exactly the union of the two subgraphs induced by the edges in $f^{-1}(\{\alpha \})$ and the edges in $f^{-1}(\{ \beta \})$.
In order to derive Conjecture \ref{conj:6normalBridgelessCubic} as a consequence of Conjecture \ref{conj:NonConflictingFlow6}, observe that the smallest counterexample to Conjecture \ref{conj:6normalBridgelessCubic} is 3-edge-connected, and, clearly, it is different from the Petersen graph. Now, if we have the nowhere zero $\mathbb{Z}_2^3$-flow $f$ from Conjecture \ref{conj:NonConflictingFlow6}, then we can view $f$ as a normal 7-edge-coloring. If we consider an edge-coloring $c$ of $G$ obtained from $f$ as follows: $c$ coincides with $f$ everywhere, except that the edges $e$ with $f(e)=\beta$ have color $c(e)=\alpha$. It is easy to see that $c$ is a normal 6-edge-coloring of $G$.

\section*{Acknowledgement} We would like to thank Robert \v{S}\'{a}mal and Jean Paul Zerafa for useful discussions over the normal colorings.



\bibliographystyle{elsarticle-num}


\end{document}